\documentclass[aps,pra,epsfigure,twocolumn,
]{revtex4}
\usepackage{amsmath}    
\usepackage{amsfonts}
\usepackage{epsfig}
\usepackage{epstopdf}

\newcommand{\miniket}[1]{\vert#1\rangle}

\newcommand{\minibra}[1]{\langle#1\vert}
\newcommand{\miniprod}[2]{\langle#1\vert#2\rangle}
\newcommand{\sand}[3]{\langle#1\vert#2\vert#3\rangle}

\newtheorem{theo}{Theorem}
\newtheorem{lemma}{Lemma}

\newenvironment{proof}[1][Proof]{\begin{trivlist}
\item[\hskip \labelsep {\bfseries #1}]}{\end{trivlist}}
\newcommand{\qed}{\nobreak \ifvmode \relax \else
      \ifdim\lastskip<1.5em \hskip-\lastskip
      \hskip1.5em plus0em minus0.5em \fi \nobreak
      \vrule height0.75em width0.5em depth0.25em\fi}

\begin{document}

\title{Alternate two-dimensional quantum walk with a single-qubit coin}

\author{C. Di Franco$^1$, M. Mc Gettrick$^2$, T. Machida$^3$ and Th. Busch$^1$}

\affiliation{$^1$ Department of Physics, University College Cork, Cork, Republic of Ireland\\
$^2$ The De Br\'un Centre for Computational Algebra, School of Mathematics, The National University of Ireland, Galway, Republic of Ireland\\
$^3$ Meiji Institute for Advanced Study of Mathematical Sciences, Meiji University, 1-1-1 Higashimita, Tamaku, Kawasaki 214-8571, Japan}

\begin{abstract}
We have recently proposed a two-dimensional quantum walk where the requirement of a higher dimensionality of the coin space is substituted with the alternance of the directions in which the walker can move [C. Di Franco, M. Mc Gettrick, and Th. Busch, Phys. Rev. Lett. {\bf 106}, 080502 (2011)]. For a particular initial state of the coin, this walk is able to perfectly reproduce the spatial probability distribution of the non-localized case of the Grover walk. Here, we present a more detailed proof of this equivalence. We also extend the analysis to other initial states, in order to provide a more complete picture of our walk. We show that this scheme outperforms the Grover walk in the generation of $x$-$y$ spatial entanglement for any initial condition, with the maximum entanglement obtained in the case of the particular aforementioned state. Finally, the equivalence is generalized to wider classes of quantum walks and a limit theorem for the alternate walk in this context is presented.
\end{abstract}

\maketitle

\section{Introduction}

The interest of a wide scientific community in the study of quantum walks~\cite{Aharonov:93,Kempe:03} has recently increased. The main reason is that, in the same way that classical random walks have applications in several fields (see, as noticeable instances, physics, computer science, economics and biology~\cite{RWGeneral}), their quantum counterparts are useful in different scenarios, from simulating quantum circuits~\cite{Childs:09} to analyzing quantum lattice gas models~\cite{feynman}. Even if the complete range of possibilities is still under investigation, interesting experimental implementations have already been realized~\cite{experiments}. Quantum walks can also efficiently generate entanglement in experimentally feasible systems~\cite{QWEntanglement}. Due to the main role of entanglement in quantum metrology, quantum computation and quantum cryptography, its generation is of fundamental importance to realize reliable devices for quantum information processing~\cite{entanglementgeneral}.

A very interesting example of a two-dimensional quantum walk is the Grover walk, which can be used to implement the two-dimensional Grover search algorithm~\cite{GroverSearch}. Unfortunately, the experimental resources required for its realization are extremely challenging for the current state-of-the-art technology. This is in general a problem common to all two-dimensional quantum walks in which the coin has to be represented by a four-level system. A simplification in this respect is highly desirable, especially from a practical perspective. Ref.~\cite{ourwalk} presents a significant step forward in this direction: we have demonstrated that the spatial probability distribution of the non-localized case of the Grover walk can be obtained using only a two-level coin and a quantum walk in alternate directions (however, the two final states will be different). The requirement of a higher dimensionality of the coin space is thus substituted by alternating the directions in which the walker can move, offering a striking advantage in terms of the experimental resources needed for its implementation. Here, we present a more complete analysis of this alternate quantum walk, considering the case of different initial conditions and providing further details of the $x$-$y$ spatial entanglement generated with respect to the Grover walk.

To analyze the asymptotic behavior of quantum walks in the long-time limit, Fourier transform methods have been used, finding limit theorems in different contexts~\cite{FourierQWs,FourierQWs2}. A similar approach has been recently exploited for the investigation of the asymptotic behavior of the entanglement in one-dimensional quantum walks~\cite{Ide}. Here, we use this tool in order to find a limit distribution for the alternate quantum walk.

The remainder of this article is organized as follows. In Sec.~\ref{alternatequantumwalk}, we describe the model under investigation. Sec.~\ref{detailedproof} presents a complete and detailed proof that the alternate quantum walk is able to perfectly reproduce the spatial probability distribution of the non-localized case of the Grover walk. In Sec.~\ref{generalization}, we study the case of different initial states of the coin, with an analysis of the $x$-$y$ spatial entanglement generated with respect to the Grover walk. Sec.~\ref{generalizedwalks} deals with the generalization of the equivalence between the quantum walks for wider classes of them and we present a limit theorem for the alternate walk in this context. Finally, Sec.~\ref{conclusions} summarizes our results.

\section{Alternate quantum walk with single-qubit coin}
\label{alternatequantumwalk}

The total state of the system considered here is a vector in the composite Hilbert space ${\cal H}={\cal H}_W\otimes{\cal H}_C$, where ${\cal H}_C$ (coin space) is a two-dimensional Hilbert space spanned by $\{\miniket{0},\miniket{1}\}$ and ${\cal H}_W$ (walker space) is an infinite-dimensional Hilbert space spanned by $\{\miniket{x,y}\}$, with $x$ and $y$ assuming all possible integer values. Let us define the basis states of this space ${\cal H}$ as $\{\miniket{x,y,c}\}$ where, for the sake of simplicity, we have defined $\miniket{x,y,c}=\miniket{x,y}_W\otimes\miniket{c}_C$. In the standard representation of two-dimensional quantum walks, $x$ and $y$ denote the position of a particle ({\it walker}) along the $x$ and $y$ directions, respectively. We stress however that here, differently from other two-dimensional quantum walks, $\miniket{c}_C$ is the state of a single-qubit coin (thus embodied by a two-level quantum system). From now on, we consider our walks (the alternate one as well as the Grover one) always starting at the origin, {\it i.e.} in $\miniket{0,0}_W$.

The evolution of the system is given by a sequence of conditional shift and coin operations. We have two different conditional shift operations
\begin{equation}
\hat{S}_x=\sum_{i,j\in \mathbb{Z}}\miniket{i-1,j,0}\minibra{i,j,0}+\sum_{i,j\in \mathbb{Z}}\miniket{i+1,j,1}\minibra{i,j,1}
\end{equation}
and
\begin{equation}
\hat{S}_y=\sum_{i,j\in \mathbb{Z}}\miniket{i,j-1,0}\minibra{i,j,0}+\sum_{i,j\in \mathbb{Z}}\miniket{i,j+1,1}\minibra{i,j,1}.
\end{equation}
If we consider the walker component $\miniket{i,j}_W$ as describing the quantized position of the walker in the $x$ and $y$ directions with increasing numbers from left to right and from bottom to top, respectively, the effect of $\hat{S}_x$ is to move the walker one step to the left (right) when the coin component is in the state $\miniket{0}_C$ ($\miniket{1}_C$) and the effect of $\hat{S}_y$ is to move the walker one step down (up) when the coin component is in the state $\miniket{0}_C$ ($\miniket{1}_C$). Our coin operation is the Hadamard gate
\begin{equation}
\hat{H}=\frac{1}{\sqrt{2}}
\begin{pmatrix}
1&1\\
1&-1
\end{pmatrix}\;,
\end{equation}
as in the original one-dimensional quantum walk~\cite{Aharonov:93}. A single time step consists here of two Hadamard operations and two movements on the $x$ and $y$ directions, according to the following sequence: coin operation, movement on $x$, coin operation, movement on $y$ (as sketched in Fig.~\ref{fig:Sketch}).

\begin{figure}[t]
\centerline{\psfig{figure=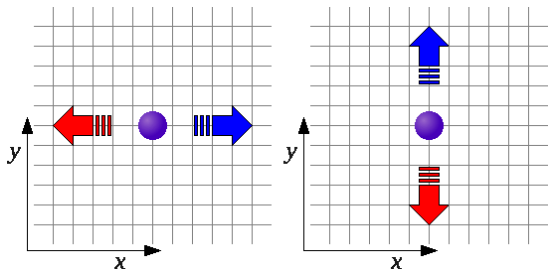,height=3.5cm}}
\caption{Sketch of the alternate quantum walk: the walker is allowed to move alternately in two orthogonal directions of a two-dimensional lattice; a coin operation is performed before each movement.}
\label{fig:Sketch}
\end{figure}

It is also useful to briefly present here the details of the Grover walk. The coin space ${\cal H}_{C'}$, in this case, is four-dimensional, so we define the basis states of the total Hilbert space as $\{\miniket{x,y,c'}\}$, where $c'\in\{0,1,2,3\}$. The states of the computational basis of the coin $\miniket{0}_{C'}$, $\miniket{1}_{C'}$, $\miniket{2}_{C'}$ and $\miniket{3}_{C'}$ correspond to movements in the left-down, left-up, right-down and right-up directions, respectively. The Grover coin operation is given by
\begin{equation}
\hat{G}=
\frac{1}{2}\left(
\begin{array}{rrrr}
-1&1&1&1\\
1&-1&1&1\\
1&1&-1&1\\
1&1&1&-1
\end{array}\right),
\end{equation}
and a single time step consists here of a Grover coin operation and a movement on the $x$-$y$ plane. In this particular scheme, the walker is always localized ({\it i.e.} the probability to find it at the origin is asymptotically larger than $0$ for $t\rightarrow\infty$), except if the coin is in the particular initial state~\cite{localization}
\begin{equation}
\frac{1}{2}(\miniket{0}_{C'}-\miniket{1}_{C'}-\miniket{2}_{C'}+\miniket{3}_{C'}).
\label{eq:GroverLocalized}
\end{equation}
The spatial probability distribution for this case can be obtained by tracing out the state of the coin, as presented in Fig.~\ref{fig:GroverWalkNonLocalized} after $t=50$ time steps.

\begin{figure}[t]
\centerline{\psfig{figure=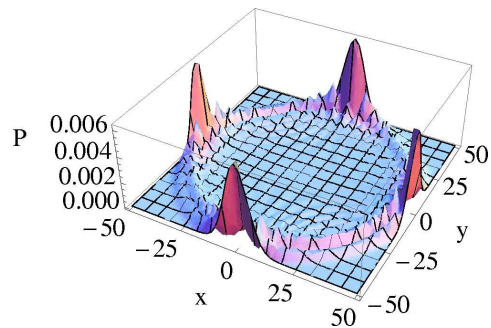,height=4.25cm}}
\caption{Spatial probability distribution after $t=50$ time steps of the two-dimensional Grover walk with the initial state of the coin as in Eq.~(\ref{eq:GroverLocalized}). Only the sites with even $x$ and $y$ are shown, as the probability is zero for all odd sites.}
\label{fig:GroverWalkNonLocalized}
\end{figure}

\section{Equivalence between the alternate quantum walk and the Grover walk}
\label{detailedproof}

In what follows, we illustrate how the coefficients of the Grover walk in the non-localized case can be mapped to the coefficients of the alternate quantum walk where the initial condition of the coin is
\begin{equation}
\frac{1}{\sqrt{2}}(\miniket{0}_C+i\miniket{1}_C),
\label{eq:AlternateInitial}
\end{equation}
as in the original symmetric one-dimensional quantum walk~\cite{Aharonov:93}. The coefficients in the decomposition of the states in the Grover walk and in the alternate quantum walk, with respect to the bases $\{\miniket{x,y,c'}\}$ and $\{\miniket{x,y,c}\}$, are defined as $\alpha_{x,y,c'}(t)$ and $\beta_{x,y,c}(t)$, respectively. It is easy to note that, for the initial states under consideration, the $\alpha_{x,y,c'}(t)$'s are real numbers, while the $\beta_{x,y,c}(t)$'s are complex numbers.

\begin{lemma}\label{lem:grover}

In the Grover walk, with the walker starting at the origin and the initial conditions
\begin{equation}
\begin{split}
\alpha_{0,0,0}(0)&=1/2,\\
\alpha_{0,0,1}(0)&=-1/2,\\
\alpha_{0,0,2}(0)&=-1/2,\\
\alpha_{0,0,3}(0)&=1/2,
\label{alpha0}
\end{split}
\end{equation}
the amplitudes satisfy the properties
\begin{eqnarray}
\nonumber&\alpha_{x-1,y,0}(t)+\alpha_{x-1,y,1}(t)\\
&+\alpha_{x+1,y,2}(t)+\alpha_{x+1,y,3}(t)=0,
\label{alpha1}\\
\nonumber&\alpha_{x,y-1,0}(t)+\alpha_{x,y-1,2}(t)\\
&+\alpha_{x,y+1,1}(t)+\alpha_{x,y+1,3}(t)=0.
\label{alpha2}
\end{eqnarray}

\end{lemma}

\begin{proof}

We proceed by induction on $t$.

\begin{description}

\item[Base Case]

It is easy to see that Eqs.~(\ref{alpha1}) and (\ref{alpha2}) are satisfied at $t=0$. For completeness, the details are as follows:
\begin{itemize}
\item At $(x,y)=(1,0)$, Eq.~(\ref{alpha1}) reads $\alpha_{0,0,0}(0)+\alpha_{0,0,1}(0)+\alpha_{2,0,2}(0)+\alpha_{2,0,3}(0)=1/2-1/2+0+0=0$.
\item At $(x,y)=(-1,0)$, Eq.~(\ref{alpha1}) reads $\alpha_{-2,0,0}(0)+\alpha_{-2,0,1}(0)+\alpha_{0,0,2}(0)+\alpha_{0,0,3}(0)=0+0-1/2+1/2=0$.
\item At $(x,y)=(0,1)$, Eq.~(\ref{alpha2}) reads $\alpha_{0,0,0}(0)+\alpha_{0,0,1}(0)+\alpha_{0,2,2}(0)+\alpha_{0,2,3}(0)=1/2-1/2+0+0=0$.
\item At $(x,y)=(0,-1)$, Eq.~(\ref{alpha2}) reads $\alpha_{0,-2,0}(0)+\alpha_{0,-2,1}(0)+\alpha_{0,0,2}(0)+\alpha_{0,0,3}(0)=0+0-1/2+1/2=0$.
\end{itemize}
For all other values of $(x,y)$, the $\alpha$'s are initially zero.

\item[Inductive Step]

Assume that both Eqs.~(\ref{alpha1}) and (\ref{alpha2}) are true [for all points $(x,y)$] at some time $t$. Then we need to prove that they hold at time $t+1$. Starting with the left-hand side of Eq.~(\ref{alpha1}), we have 
\begin{widetext}
\begin{equation}
\begin{split}
&\alpha_{x-1,y,0}(t+1)+\alpha_{x-1,y,1}(t+1)+\alpha_{x+1,y,2}(t+1)+\alpha_{x+1,y,3}(t+1)=\sum_{j=0}^3[G_{0j}\alpha_{x,y+1,j}(t)+G_{1j}\alpha_{x,y-1,j}(t)\\
&+G_{2j}\alpha_{x,y+1,j}(t)+G_{3j}\alpha_{x,y-1,j}(t)]=\sum_{j=0}^3[(G_{0j} + G_{2j})\alpha_{x,y+1,j}(t)+(G_{1j}+G_{3j})\alpha_{x,y-1,j}(t)]\\
&=\alpha_{x,y+1,1}(t)+\alpha_{x,y+1,3}(t)+\alpha_{x,y-1,0}(t)+\alpha_{x,y-1,2}(t),
\end{split}
\end{equation}
\end{widetext}
which is identically zero because we have assumed Eq.~(\ref{alpha2}) true at time $t$. Here, $G_{ij}$ $(i,j=0,1,2,3)$ is the element of the matrix $\hat{G}$ corresponding to $\miniket{i}_C\minibra{j}$. It is straightforward to proceed in the same way to prove Eq.~(\ref{alpha2}) at time $t+1$, assuming that Eq.~(\ref{alpha1}) is true at time $t$.

\end{description}

\hfill\qed

\end{proof}

\begin{theo}
\label{theo1}

The relations between the amplitudes $\beta_{x,y,j}(t)$ of the alternate quantum walk with the initial conditions
\begin{equation}
\begin{split}
\beta_{0,0,0}(0)&=1/\sqrt{2},\\
\beta_{0,0,1}(0)&=i/\sqrt{2}
\end{split}
\end{equation}
and the amplitudes $\alpha_{x,y,k}(t)$ of the Grover walk with the initial state as in Eq.~(\ref{alpha0}) are given by
\begin{eqnarray}
\beta_{x,y,0}(t)&=&(-1)^te^{i\pi/4}[\alpha_{x,y,0}(t)+i\alpha_{x,y,2}(t)]\label{beta1},
\label{theoremequation1}\\
\beta_{x,y,1}(t)&=&(-1)^te^{i\pi/4}[-\alpha_{x,y,1}(t)+i\alpha_{x,y,3}(t)\label{beta2}].
\label{theoremequation2}
\end{eqnarray}

\end{theo}

\begin{proof}

We proceed by induction on $t$.

\begin{description}

\item[Base Case] At $t=0$, all the amplitudes are zero outside of the origin $(x,y)=(0,0)$. At the origin, we have
\begin{eqnarray}
\begin{split}
\beta_{0,0,0}(0)&=(-1)^0e^{i\pi/4}[\alpha_{0,0,0}(0)+i\alpha_{0,0,2}(0)]\\
&=\frac{1+i}{\sqrt{2}}\left(\frac{1}{2}-\frac{i}{2}\right)=\frac{1}{\sqrt{2}},
\end{split}\\
\begin{split}
\beta_{0,0,1}(0)&=(-1)^0e^{i\pi/4}[-\alpha_{0,0,1}(0)+i\alpha_{0,0,3}(0)]\\
&=\frac{1+i}{\sqrt{2}}\left(\frac{1}{2}+\frac{i}{2}\right)=\frac{i}{\sqrt{2}}.
\end{split}
\end{eqnarray}

\item[Inductive Step]

Assume that both Eqs.~(\ref{beta1}) and (\ref{beta2}) are true [for all points $(x,y)$] at some time $t$. Then, we need to prove that they hold at time $t+1$. The progression of each walk is as follows.

\begin{description}

\item[Alternate quantum walk]

\begin{eqnarray}
\begin{split}
\beta_{x,y,0}(t+1)=&\,\frac{1}{2}[\beta_{x+1,y+1,0}(t)+\beta_{x+1,y+1,1}(t)\\
&+\beta_{x-1,y+1,0}(t)-\beta_{x-1,y+1,1}(t)],
\label{alt1}
\end{split}\\
\begin{split}
\beta_{x,y,1}(t+1)=&\,\frac{1}{2}[\beta_{x+1,y-1,0}(t)+\beta_{x+1,y-1,1}(t)\\
&-\beta_{x-1,y-1,0}(t)+\beta_{x-1,y-1,1}(t)].
\label{alt2}
\end{split}
\end{eqnarray}

\item[Grover walk]

\begin{equation}
\begin{split}
\alpha_{x,y,0}(t+1)&=\sum_{j=0}^3 G_{0j}\alpha_{x+1,y+1,j}(t),\\
\alpha_{x,y,1}(t+1)&=\sum_{j=0}^3 G_{1j}\alpha_{x+1,y-1,j}(t),\\
\alpha_{x,y,2}(t+1)&=\sum_{j=0}^3 G_{2j}\alpha_{x-1,y+1,j}(t),\\
\alpha_{x,y,3}(t+1)&=\sum_{j=0}^3 G_{3j}\alpha_{x-1,y-1,j}(t).
\end{split}
\end{equation}

\end{description}

Starting from Eq.~(\ref{alt1}), we have

\begin{equation}
\begin{split}
\beta_{x,y,0}(t+1)&=\frac{1}{2}(-1)^te^{i\pi/4}\{\alpha_{x+1,y+1,0}(t)\!-\!\alpha_{x+1,y+1,1}(t)\\
&+i\alpha_{x-1,y+1,2}(t)-i\alpha_{x-1,y+1,3}(t)\\
&+[\alpha_{x-1,y+1,0}(t)+\alpha_{x-1,y+1,1}(t)\\
&+i\alpha_{x+1,y+1,2}(t)+i\alpha_{x+1,y+1,3}(t)]\}.
\label{ind1}
\end{split}
\end{equation}
Now we use Lemma~\ref{lem:grover}, specifically the relation
\begin{equation}
\alpha_{x-1,y,0}(t)\!+\!\alpha_{x-1,y,1}(t)\!=\!-\alpha_{x+1,y,2}(t)\!-\!\alpha_{x+1,y,3}(t),
\end{equation}
to replace the terms within the square brackets in Eq.~(\ref{ind1}) and obtain
\begin{widetext}
\begin{equation}
\begin{split}
\beta_{x,y,0}(t+1)&=\frac{1}{2}(-1)^te^{i\pi/4}\{\alpha_{x+1,y+1,0}(t)-\alpha_{x+1,y+1,1}(t)+i\alpha_{x-1,y+1,2}(t)-i\alpha_{x-1,y+1,3}(t)\\
&+[-\alpha_{x+1,y+1,2}(t)-\alpha_{x+1,y+1,3}(t)-i\alpha_{x-1,y+1,0}(t)-i\alpha_{x-1,y+1,1}(t)]\}\\
&=\frac{1}{2}(-1)^{t+1}e^{i\pi/4}\{-\alpha_{x+1,y+1,0}(t)+\alpha_{x+1,y+1,1}(t)+\alpha_{x+1,y+1,2}(t)+\alpha_{x+1,y+1,3}(t)\\
&+i\alpha_{x-1,y+1,0}(t)+i\alpha_{x-1,y+1,1}(t)-i\alpha_{x-1,y+1,2}(t)+i\alpha_{x-1,y+1,3}(t)\}\\
&=(-1)^{t+1}e^{i\pi/4}\{\sum_{j=0}^3 G_{0j}\alpha_{x+1,y+1,j}(t)+i\sum_{j=0}^3G_{2j}\alpha_{x-1,y+1,j}(t)\}\\
&=(-1)^{t+1}e^{i\pi/4}\{\alpha_{x,y,0}(t+1)+i\alpha_{x,y,2}(t+1)\},
\end{split}
\end{equation}
\end{widetext}
which completes the proof for Eq.~(\ref{beta1}). An analogous analysis allows us to prove the partner Eq.~(\ref{beta2}).

\end{description}
\hfill\qed

\end{proof}

It is now straightforward to see why the two walks have the same spatial probability distribution. The probability to find the walker at position $(x,y)$ after $t$ time steps is
\begin{equation}
P(x,y)=\sum_{j=0}^3|\alpha_{x,y,j}(t)|^2
\end{equation}
for the Grover walk and
\begin{equation}
P(x,y)=\sum_{j=0}^1|\beta_{x,y,j}(t)|^2
\end{equation}
for the alternate quantum walk. The relations in Eqs.~(\ref{theoremequation1}) and (\ref{theoremequation2}) and the fact that the $\alpha_{x,y,j}(t)$'s are real numbers guarantee that the two probabilities are the same for all the points $(x,y)$ at any time $t$.

\section{Alternate quantum walk with different initial conditions}
\label{generalization}

So far, the analysis has been performed on the alternate quantum walk for an initial state of the coin given by Eq.~(\ref{eq:AlternateInitial}). Clearly, any two-level state can be chosen as the initial condition of the coin. If we focus on pure states, its general form can be cast as
\begin{equation}
\miniket{\psi}=\cos\frac{\theta}{2}\miniket{0}+e^{i\phi}\sin\frac{\theta}{2}\miniket{1},
\label{generalinitialstate}
\end{equation}
with $\theta\in[0,\pi]$ and $\phi\in[0,2\pi]$. Let us first analyze the spatial probability distribution for particular instances of the initial conditions. If we consider the states
\begin{equation}
\miniket{\psi_1}=\miniket{1},
\label{initialstate1}
\end{equation}
and
\begin{equation}
\miniket{\psi_2}=\frac{1}{\sqrt{2}}(\miniket{0}-\miniket{1})
\label{initialstate2}
\end{equation}
the corresponding spatial probability distributions after $t=50$ time steps are those shown in Fig.~\ref{plot1b} and Fig.~\ref{plot1c},  respectively.
\begin{figure}[t]
\centerline{\psfig{figure=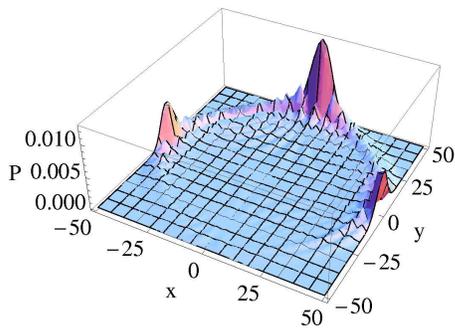,height=4.25cm}}
\caption{Spatial probability distribution after $t=50$ time steps of the alternate quantum walk with the initial state of the coin as in Eq.~\eqref{initialstate1}.}
\label{plot1b}
\end{figure}
\begin{figure}[t]
\centerline{\psfig{figure=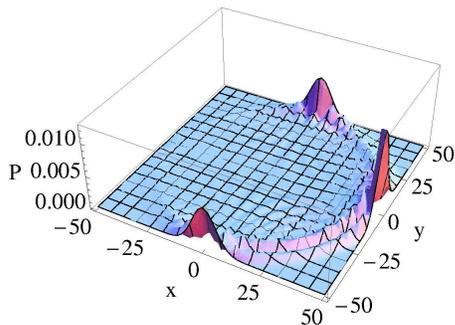,height=4.25cm}}
\caption{Spatial probability distribution after $t=50$ time steps of the alternate quantum walk with the initial state of the coin as in Eq.~\eqref{initialstate2}.}
\label{plot1c}
\end{figure}
The state orthogonal to $\miniket{\psi_1}$ ($\miniket{\psi_2}$) gives a distribution that is the symmetric to Fig.~\ref{plot1b} (Fig.~\ref{plot1c}) with respect to the $x$ ($y$) axis. It is easy to notice that the spatial probability distribution is enhanced in a particular direction of the $x$-$y$ plane on which the walker is moving. More precisely, through a numerical study, we have found a correspondence between the direction of the Bloch vector~\cite{nielsen} of the initial coin state and this enhanced direction. In fact, we just need to consider the projection of this vector onto the azimuthal plane of the Bloch sphere orthogonal to the vector corresponding to the state in Eq.~\eqref{eq:AlternateInitial}. The states in Eq.~\eqref{initialstate1}, Eq.~\eqref{initialstate2} and those orthogonal to them belong to this plane and each one corresponds to a specific direction of the $x$-$y$ plane on which the walker is moving [for instance, the state in Eq.~\eqref{initialstate1} corresponds to the positive $y$ direction]. The closer the projection of the initial state of the coin is to one of these four states, the more the corresponding direction is enhanced. Clearly, the only states on the Bloch sphere that have a null projection on the aforementioned plane are the one in Eq.~\eqref{eq:AlternateInitial} and its orthogonal state. These are the only cases in which the probability distribution is symmetric with respect to both the axes (see Fig.~\ref{fig:GroverWalkNonLocalized}). 

Let us now focus the analysis on the generation of entanglement in our scheme. In Ref.~\cite{ourwalk}, we have already shown that the alternate quantum walk with the initial condition as in Eq.~\eqref{eq:AlternateInitial} is able to generate more $x$-$y$ spatial entanglement than the Grover walk in its non-localized case. Here, we extend the investigation to other initial states of the coin. In order to evaluate the $x$-$y$ spatial entanglement, we need to trace out the degree of freedom embodied by the coin, and then calculate the negativity $N$ of the partial transpose, in its generalization for higher-dimensional systems (so as to have $0\le N\le 1$)~\cite{quditsnegativity}. Also in this case, the initial state of the coin in the alternate quantum walk is described by Eq.~\eqref{generalinitialstate} and we have considered a number of steps $t=10$, so as to perform the numerical calculations in a reasonable time. However, we have checked also for longer times and the behavior does not change. We first report the results obtained by fixing the value of $\phi$ and varying $\theta$. We have considered $\theta$ in the interval $[0,\pi]$ and uniformly taken $60$ points in it. The $x$-$y$ spatial entanglement generated in the alternate quantum walk is presented in Fig.~\ref{plot2a}, Fig.~\ref{plot2b} and Fig.~\ref{plot2d}, for a value of $\phi$ equal to $0$, $\pi/8$ and $\pi/2$, respectively.
\begin{figure}[t]
\centerline{\psfig{figure=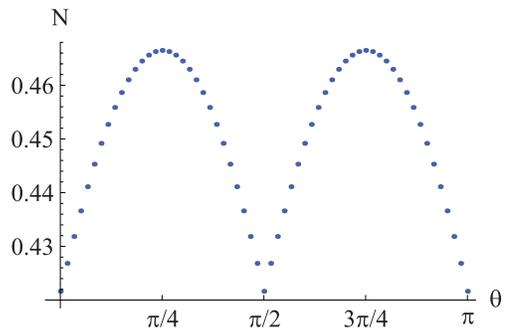,height=4.25cm}}
\caption{Entanglement between $x$ and $y$ position of the walker in the alternate quantum walk, with the initial conditions of the coin as in Eq.~\eqref{generalinitialstate}, $\phi=0$ and after a number of steps $t=10$.}
\label{plot2a}
\end{figure}
\begin{figure}[t]
\centerline{\psfig{figure=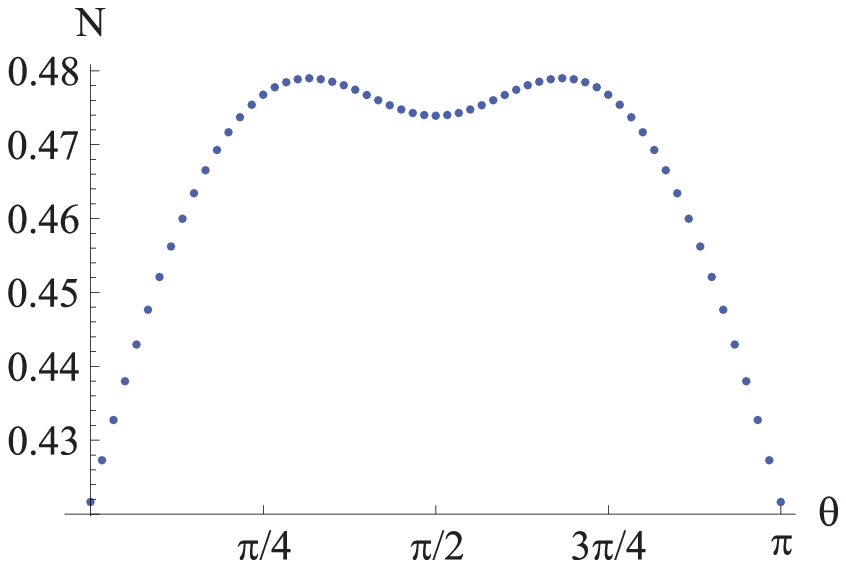,height=4.25cm}}
\caption{Entanglement between $x$ and $y$ position of the walker in the alternate quantum walk, with the initial conditions of the coin as in Eq.~\eqref{generalinitialstate}, $\phi=\pi/8$ and after a number of steps $t=10$.}
\label{plot2b}
\end{figure}
\begin{figure}[t]
\centerline{\psfig{figure=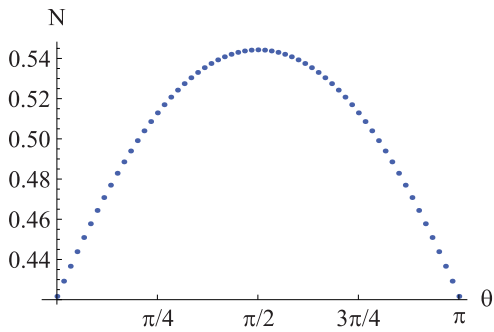,height=4.25cm}}
\caption{Entanglement between $x$ and $y$ position of the walker in the alternate quantum walk, with the initial conditions of the coin as in Eq.~\eqref{generalinitialstate}, $\phi=\pi/2$ and after a number of steps $t=10$.}
\label{plot2d}
\end{figure}

In order to provide a more complete picture of the generated $x$-$y$ spatial entanglement, we have also studied its behavior against both $\theta$ and $\phi$. With respect to $\phi$, the results are periodic with a period of $\pi$. The computational power required for the calculations is clearly larger in this case, and we have therefore reduced the number of points in each interval for $\theta$ and $\phi$. As a good compromise between the required computational time and the readability of the plot, we have chosen to take $20$ points for each interval. The result is presented in Fig.~\ref{plot3}.
\vskip0.25in
\begin{figure}[t]
\centerline{\psfig{figure=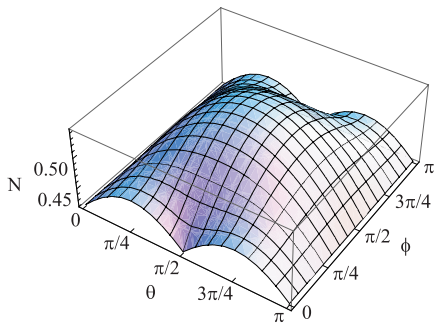,height=4.25cm}}
\caption{Entanglement between $x$ and $y$ position of the walker in the alternate quantum walk against both $\theta$ and $\phi$, with the initial conditions of the coin as in Eq.~\eqref{generalinitialstate} and after a number of steps $t=10$.}
\label{plot3}
\end{figure}

The generated entanglement is minimal for the values of $\theta$ and $\phi$ corresponding to the states in Eq.~\eqref{initialstate1}, Eq.~\eqref{initialstate2} and those orthogonal to them (with a negativity, for a number of steps  $t=10$, of $N \sim 0.42164$), while it is maximal for the state in Eq.~\eqref{eq:AlternateInitial} and the orthogonal one ($N \sim 0.54428$). We have also checked the minimum value of the entanglement generated in the alternate quantum walk [with the initial condition of the coin as in Eq.~\eqref{initialstate1}] with respect to the entanglement generated in the non-localized case of the Grover walk (the non-localized case gives the maximum value of the generated entanglement for this walk) after the same number of time steps and we have found that the former is always larger than the latter. We have investigated this property for larger numbers of steps, up to $t=20$, always obtaining the same result. We can thus state that, in all the considered cases, our scheme outperforms the Grover walk in the generation of $x$-$y$ spatial entanglement for any initial condition of the coin.

\section{Generalization to wider classes of quantum walks and limit theorem for the alternate walk}
\label{generalizedwalks}

The equivalence proven in Sec.~\ref{detailedproof} can be extended to a wider scenario. In this case, the generalized Grover operation is given by
\begin{equation}
\hat{A}=
\left(
\begin{array}{rrrr}
-c^2&|cs|&|cs|&s^2\\
|cs|&-s^2&c^2&|cs|\\
|cs|&c^2&-s^2&|cs|\\
s^2&|cs|&|cs|&-c^2
\end{array}\right),
\end{equation}
with $c=\cos\gamma$, $s=\sin\gamma$ and $\gamma\in (0,2\pi)$ ($\gamma\ne\pi/2,\pi,3\pi/2$). Let us take the initial state of the coin in the generalized Grover walk as
\begin{equation}
q_0\miniket{0}_{C'}+q_1\miniket{1}_{C'}+q_2\miniket{2}_{C'}+q_3\miniket{3}_{C'},
\end{equation}
with $|q_0|^2+|q_1|^2+|q_2|^2+|q_3|^2=1$. For the generalized alternate quantum walk, we consider the coin operation 
\begin{equation}
\hat{U}=
\begin{pmatrix}
c&s\\
s&-c
\end{pmatrix}\;,
\end{equation}
and the initial state of the coin given by
\begin{equation}
\nu_0\miniket{0}_{C}+\nu_1\miniket{1}_{C},
\label{generalizedwalkinitialstate}
\end{equation}
with $|\nu_0|^2+|\nu_1|^2=1$. In this case, we can prove the following theorem.

\begin{theo}

If we take the initial states of the coins satisfying the conditions
\begin{equation}
\begin{split}
|\nu_0|^2&=|\nu_1|^2,\\
\nu_0\nu_1^*+&\nu_0^*\nu_1=0
\label{machidaproof1}
\end{split}
\end{equation}
and
\begin{equation}
\begin{split}
q_0=q_3&=(-1)^\xi\frac{|cs|}{\sqrt{2}s},\\
q_1=q_2&=-(-1)^\xi\frac{s}{\sqrt{2}}
\end{split}
\end{equation}
with $\xi=0,1$, and the walker starting at the origin in both cases, then the generalized alternate quantum walk and the generalized Grover walk have the same spatial probability distribution at any time $t$.

\end{theo}

The proof is along the line of Th.~\ref{theo1} in Sec.~\ref{detailedproof}. Eq.~\eqref{machidaproof1} can be rewritten as
\begin{equation}
\begin{split}
&|\nu_0|=\frac{1}{\sqrt{2}},\\
&\nu_1=(-1)^\kappa i\nu_0
\end{split}
\end{equation}
with $\kappa=0,1$ [corresponding to the state in Eq.~\eqref{eq:AlternateInitial} and the orthogonal one, respectively]. In this case, Eqs.~\eqref{theoremequation1} and \eqref{theoremequation2} are substituted by
\begin{equation}
\begin{split}
\beta_{x,y,0}&=(-1)^{t+\xi}\sqrt{2}\,\nu_0\{c+(-1)^\kappa i\,s\}\\
&\times[sign(cs)\alpha_{x,y,0}(t)+(-1)^\kappa i\,\alpha_{x,y,2}(t)],\\
\beta_{x,y,1}&=(-1)^{t+\xi}\sqrt{2}\,\nu_0\{c+(-1)^\kappa i\,s\}\\
&\times[-\alpha_{x,y,1}(t)+sign(cs)(-1)^\kappa i\,\alpha_{x,y,3}(t)],
\end{split}
\end{equation}
where $sign(x)$ denotes the sign of $x$, while
\begin{equation}
\begin{split}
&|s|\alpha_{x-1,y,0}(t)+|c|\alpha_{x-1,y,1}(t)\\
&+|c|\alpha_{x+1,y,2}(t)+|s|\alpha_{x+1,y,3}(t)=0,\\
&|s|\alpha_{x,y-1,0}(t)+|c|\alpha_{x,y-1,2}(t)\\
&+|c|\alpha_{x,y+1,1}(t)+|s|\alpha_{x,y+1,3}(t)=0
\end{split}
\end{equation}
take the place of Eqs.~\eqref{alpha1} and ~\eqref{alpha2}, when the walk starts at the origin and the initial state corresponds to
\begin{equation}
\begin{split}
&\alpha_{0,0,0}(0)=q_0=(-1)^\xi\frac{|cs|}{\sqrt{2}s},\\
&\alpha_{0,0,1}(0)=q_1=-(-1)^\xi\frac{s}{\sqrt{2}},\\
&\alpha_{0,0,2}(0)=q_2=-(-1)^\xi\frac{s}{\sqrt{2}},\\
&\alpha_{0,0,3}(0)=q_3=(-1)^\xi\frac{|cs|}{\sqrt{2}s}.\\
\end{split}
\end{equation}
The proof is a straightforward generalization of the one in Sec.~\ref{detailedproof} so, for the sake of simplicity, we do not report it here entirely.

We now want to find a limit theorem for the alternate quantum walk, in the same way as limit theorems have been already found for other quantum walks in the literature~\cite{FourierQWs,FourierQWs2}. For this, we define $\miniket{\psi_t(x,y)}$ as
\begin{equation}
\miniket{\psi_t(x,y)}=\beta_{x,y,0}\miniket{0}_C+\beta_{x,y,1}\miniket{1}_C
\end{equation}
and we put $\miniket{0}_C=\,^T[1,0]$,  $\miniket{1}_C=\,^T[0,1]$, where $T$ is the transpose operator. The probability that the quantum walker $(X_t,Y_t)$ is at position $(x,y)\in\mathbb{Z}^2$ at time $t$ is
\begin{equation}
P[(X_t,Y_t)=(x,y)]=\miniprod{\psi_t(x,y)}{\psi_t(x,y)}.
\end{equation}
The Fourier transform $\miniket{\hat\Psi_t(k_x,k_y)}$ of $\miniket{\psi_t(x,y)}$, with $k_x,k_y\in[-\pi,\pi)$, is given by
\begin{equation}
\miniket{\hat\Psi_t(k_x,k_y)}=\sum_{(x,y)\in\mathbb{Z}^2}e^{-i(k_x x+k_y y)}\miniket{\psi_t(x,y)}.
\label{fourier}
\end{equation}
By the inverse Fourier transform, we have
\begin{equation}
\miniket{\psi_t(x,y)}=\int_{-\pi}^\pi\frac{dk_x}{2\pi}\int_{-\pi}^\pi\frac{dk_y}{2\pi}e^{i(k_x x+k_y y)}\miniket{\hat\Psi_t(k_x,k_y)}.
\end{equation}
From the time evolution of the alternate quantum walk and Eq.~\eqref{fourier}, the Fourier transform satisfies
\begin{equation}
\miniket{\hat\Psi_{t+1}(k_x,k_y)}=\hat{V}(k_x,k_y)\miniket{\hat\Psi_t(k_x,k_y)},
\end{equation}
where $\hat{V}(k_x,k_y)=\hat{R}(k_y)\hat{U}\hat{R}(k_x)\hat{U}$ and
\begin{equation}
\hat{R}(k)=
\begin{pmatrix}
e^{ik}&0\\
0&e^{-ik}
\end{pmatrix}\;.
\end{equation}
Therefore, we can get
\begin{equation}
\miniket{\hat\Psi_t(k_x,k_y)}=\hat{V}(k_x,k_y)^t\miniket{\hat\Psi_0(k_x,k_y)}.
\label{hatPsit}
\end{equation}
From Eq.~\eqref{generalizedwalkinitialstate}, the initial state becomes
\begin{equation}
\miniket{\psi_0(x,y)}=
\left\{\begin{array}{cc}
^T[\nu_0,\nu_1] & \hskip0.5cm\text{for}\hskip0.25cm(x,y)=(0,0),\\
^T[0,0] & \hskip0.5cm\text{for}\hskip0.25cm(x,y)\ne(0,0).
\end{array}\right.
\label{psi0}
\end{equation}
Let us note that $\miniket{\hat{\Psi}_0(k_x,k_y)}=\miniket{\psi_0(0,0)}$.

We can now prove the following theorem.

\begin{theo}

As $t\rightarrow\infty$, we can obtain the limit distribution for the alternate quantum walk as
\begin{equation}
\lim_{t\rightarrow\infty}P\left(\frac{X_t}{t}\le x,\frac{Y_t}{t}\le y\right)=\int_{-\infty}^x du\int_{-\infty}^y dv\,f(u,v),
\end{equation}
where
\begin{equation}
\begin{split}
f(x,y)&=\frac{1}{\pi^2(1-x^2)(1-y^2)}\{1-(|\nu_0|^2-|\nu_1|^2)y\\
&-\frac{\nu_0\nu_1^*+\nu_0^*\nu_1}{2cs}[c^2(x-y)+s^2(x+y)]\}I_D(x,y),\\
D&=\left\{(x,y)\left|\frac{(x+y)^2}{4c^2}+\frac{(x-y)^2}{4s^2}<1\right.\right\},
\end{split}
\end{equation}
and $I_D(x,y)=1$ if $(x,y)\in D$, $I_D(x,y)=0$ if $(x,y)\notin D$.

\end{theo}

\begin{proof}

\begin{description}

Our approach is based on the Fourier analysis applied to quantum walks on $\mathbb{Z}^2$~\cite{FourierQWs,FourierQWs2}. We concentrate on the characteristic function $E(e^{i(z_1X_t/t+z_2Y_t/t)})$ as $t\rightarrow\infty$, where $E(X)$ denotes the expected value of $X$. The eigenvalues $\lambda_j(k_x,k_y)$ (with $j=1,2$) of $\hat{V}(k_x,k_y)$ are
\begin{equation}
\begin{split}
&\lambda_j(k_x,k_y)=c^2\cos(k_x+k_y)+s^2\cos(k_x-k_y)\\
&+i(-1)^j\sqrt{1-[c^2\cos(k_x+k_y)+s^2\cos(k_x-k_y)]^2}.
\end{split}
\end{equation}
The normalized eigenvector $\miniket{v_j(k_x,k_y)}$, corresponding to $\lambda_j(k_x,k_y)$, is given by
\begin{equation}
\begin{split}&\miniket{v_j(k_x,k_y)}=\frac{1}{\sqrt{N_j(k_x,k_y)}}\\
&\times\begin{pmatrix}
cs[e^{i(k_x+k_y)}-e^{-i(k_x-k_y)}]\\
i\left[g_1(k_x,k_y)+(-1)^j\sqrt{1-g_2(k_x,k_y)^2}\right]
\end{pmatrix}\;,
\end{split}
\end{equation}
where
\begin{equation}
\begin{split}
&g_1(k_x,k_y)=-c^2\sin(k_x+k_y)+s^2\sin(k_x-k_y),\\
&g_2(k_x,k_y)=c^2\cos(k_x+k_y)+s^2\cos(k_x-k_y),
\end{split}
\end{equation}
and $N_j(k_x,k_y)$ is a normalization factor. Eq.~\eqref{hatPsit} can be written as
\begin{equation}
\begin{split}
&\miniket{\hat\Psi_t(k_x,k_y)}=\hat{V}(k_x,k_y)^t\miniket{\hat\Psi_0(k_x,k_y)}\\
&=\sum_{j=1}^2\lambda_j(k_x,k_y)^t\sand{v_j(k_x,k_y)}{\hat\Psi_0(k_x,k_y)}{v_j(k_x,k_y)}.
\end{split}
\end{equation}
Note that $\miniket{\hat\Psi_0(k_x,k_y)}=\miniket{\psi_0(0,0)}=\,^T[\nu_0,\nu_1]$ [see Eq.~\eqref{psi0}]. The joint moment of $X_t$ and $Y_t$ is expressed as
\begin{equation}
\begin{split}
&E(X_t^{r_1}Y_t^{r_2})=\sum_{(x,y)\in\mathbb{Z}^2} x^{r_1}y^{r_2}P[(X_t,Y_t)=(x,y)]\\
&=\int_{-\pi}^\pi\frac{dk_x}{2\pi}\int_{-\pi}^\pi\frac{dk_y}{2\pi}\sand{\hat\Psi_t(k_x,k_y)}{D_x^{r_1}D_y^{r_2}}{\hat\Psi_t(k_x,k_y)}\\
&=(t)_{r_1+r_2}\int_{-\pi}^\pi\frac{dk_x}{2\pi}\int_{-\pi}^\pi\frac{dk_y}{2\pi}\sum_{j=1}^2\left[\frac{D_x\lambda_j(k_x,k_y)}{\lambda_j(k_x,k_y)}\right]^{r_1}\\
&\times\left[\frac{D_y\lambda_j(k_x,k_y)}{\lambda_j(k_x,k_y)}\right]^{r_2}|\miniprod{v_j(k_x,k_y)}{\hat\Psi_0(k_x,k_y)}|^2\\
&+O(t^{r_1+r_2-1})
\end{split}
\end{equation}
with $D_x=i(\partial{}/\partial{k_x})$, $D_y=i(\partial{}/\partial{k_y})$ and $(t)_r=t(t-1)\times\cdot\cdot\cdot\times(t-r+1)$. Noting that
\begin{equation}
\begin{split}
&\frac{D_x\lambda_j(k)}{\lambda_j(k)}=\frac{-(-1)^j[c^2\sin(k_x+k_y)+s^2\sin(k_x-k_y)]}{\sqrt{1-[c^2\cos(k_x+k_y)+s^2\cos(k_x-k_y)]^2}},\\
&\frac{D_y\lambda_j(k)}{\lambda_j(k)}=\frac{-(-1)^j[c^2\sin(k_x+k_y)-s^2\sin(k_x-k_y)]}{\sqrt{1-[c^2\cos(k_x+k_y)+s^2\cos(k_x-k_y)]^2}},
\end{split}
\end{equation}
we see that
\begin{equation}
\begin{split}
&\lim_{t\rightarrow\infty}E\left[\left(\frac{X_t}{t}\right)^{r_1}\left(\frac{Y_t}{t}\right)^{r_2}\right]\\
&=\int_{-\pi}^\pi\frac{dk_x}{2\pi}\int_{-\pi}^\pi\frac{dk_y}{2\pi}\sum_{j=1}^2\left[\frac{D_x\lambda_j(k_x,k_y)}{\lambda_j(k_x,k_y)}\right]^{r_1}\\
&\times\left[\frac{D_y\lambda_j(k_x,k_y)}{\lambda_j(k_x,k_y)}\right]^{r_2}|\miniprod{v_j(k_x,k_y)}{\hat\Psi_0(k_x,k_y)}|^2\\
&=\int_{-\infty}^\infty dx\int_{-\infty}^\infty dy\,x^{r_1}y^{r_2}f(x,y),
\end{split}
\label{limE}
\end{equation}
where
\begin{equation}
\begin{split}
&f(x,y)=\frac{1}{\pi^2(1-x^2)(1-y^2)}\{1-(|\nu_0|^2-|\nu_1|^2)y\\
&-\frac{\nu_0\nu_1^*+\nu_0^*\nu_1}{2cs}[c^2(x-y)+s^2(x+y)]\}I_D(x,y),
\end{split}
\end{equation}
with
\begin{equation}
D=\left\{(x,y)\left|\frac{(x+y)^2}{4c^2}+\frac{(x-y)^2}{4s^2}<1\right.\right\}.
\end{equation}
By Eq.~\eqref{limE}, we can compute the characteristic function $E(e^{i(z_1X_t/t+z_2Y_t/t)})$ as $t\rightarrow\infty$. Since $f(x,y)$ is a density function (see Ref.~\cite{FourierQWs2}), proof is completed.

\end{description}
\hfill\qed

\end{proof}

For the sake of completeness, we present the limit density function $f(x,y)$ with $\gamma=\pi/4$ and $\miniket{\psi_0(0,0)}=\,^T[1/\sqrt{2},i/\sqrt{2}]$, corresponding to the particular case of the alternate quantum walk studied in Ref.~\cite{ourwalk}, in Fig.~\ref{limitfunction}.

\begin{figure}[t]
\centerline{\psfig{figure=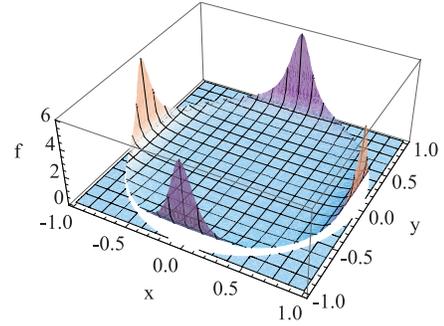,height=4.25cm}}
\caption{Limit density function $f(x,y)$ for the alternate quantum walk with $\gamma=\pi/4$ and $\miniket{\psi_0(0,0)}=\,^T[1/\sqrt{2},i/\sqrt{2}]$.}
\label{limitfunction}
\end{figure}

\section{Conclusions}
\label{conclusions}
We have provided a detailed proof of the equivalence between the spatial probability distributions of the alternate quantum walk proposed in Ref.~\cite{ourwalk} and the Grover walk. A deeper investigation of the alternate quantum walk has been performed, considering different initial conditions, and we have found a correspondence between the Bloch vector of the initial coin state and an enhancement of the spatial probability distribution in a particular direction of the plane on which the walker is moving. In the context of general initial conditions, we have also investigated the generation of $x$-$y$ spatial entanglement; we have proved that its maximum corresponds to the initial conditions studied in Ref.~\cite{ourwalk}. The equivalence between quantum walks with two- and four-dimensional coins has then been extended to wider classes and a limit theorem has been put forward for this generalized alternate quantum walk. We believe that the extensive exploration of the scheme provided in this work will contribute to the development of interesting proposals for the exploitation of quantum walks in feasible experimental settings.

\section{Acknowledgments}
We thank C. Gillis, N. Lo Gullo, M. Paternostro and Y. Shikano for discussions. We acknowledge support from Science Foundation Ireland under Grants No. 05/IN/I852 and No. 10/IN.1/I2979. C.D.F. is supported by the Irish Research Council for Science, Engineering and Technology.

\end{document}